\documentclass[12pt]{article}
\usepackage{times}
\usepackage{amsfonts}
\usepackage{amsmath}
\usepackage{enumitem}
\usepackage{mathtools}
\usepackage{centernot}
\usepackage{soul}
\usepackage{booktabs}
\usepackage{latexsym}
\usepackage[ruled]{algorithm2e}
\SetAlFnt{\small}
\SetAlCapFnt{\small}
\SetAlCapNameFnt{\small}
\SetAlCapHSkip{0pt}
\IncMargin{-\parindent}
\usepackage{algorithmic}
\usepackage{nicefrac}
\usepackage[
  bookmarks=true,
  bookmarksnumbered=true,
  bookmarksopen=true,
  pdfborder={0 0 0},
  breaklinks=true,
  colorlinks=true,
  linkcolor=black,
  citecolor=black,
  filecolor=black,
  urlcolor=black,
]{hyperref}
\usepackage{graphicx}
\usepackage{xcolor}
\usepackage[margin=1in]{geometry}
\usepackage{amssymb}
\usepackage{amsmath}
\usepackage[square]{natbib}
\setcitestyle{numbers}
\usepackage{amsthm}
\usepackage[capitalize]{cleveref}
\usepackage{bbm}
\usepackage[normalem]{ulem}
\usepackage{tikz}
\usetikzlibrary{decorations.pathreplacing,angles,quotes}
\usepackage{MnSymbol}
\newtheorem{theorem}{Theorem}[section]

\newtheorem{proposition}[theorem]{Proposition}

\newtheorem{example}[theorem]{Example}

\newtheorem*{question*}{Question}
\DeclarePairedDelimiter\parentheses{(}{)}
\DeclarePairedDelimiter\braces{\{}{\}}
\DeclarePairedDelimiter\brackets{[}{]}
\DeclarePairedDelimiter\absolute{|}{|}

\DeclareMathOperator{\supp}{supp}

\begin{document}
\title{Information Design in the Principal-Agent Problem}
\date{September 4, 2024}
\author{Yakov Babichenko\thanks{Technion--Israel Institute of Technology | \emph{E-mail}: \href{mailto:yakovbab@technion.ac.il}{yakovbab@technion.ac.il}.} 
\and 
Inbal Talgam-Cohen\thanks{Tel Aviv University and Technion--Israel Institute of Technology | \emph{E-mail}: \href{mailto:inbaltalgam@gmail.com}{inbaltalgam@gmail.com}.}
\and 
Haifeng Xu \thanks{University of Chicago | \emph{E-mail}: \href{mailto:haifengxu@uchicago.edu}{haifengxu@uchicago.edu}.} \and Konstantin Zabarnyi\thanks{Technion--Israel Institute of Technology | \emph{E-mail}: \href{mailto:konstzab@gmail.com}{konstzab@gmail.com}.}}
\maketitle

\begin{abstract}
We study a variant of the principal-agent problem in which the principal does not directly observe the agent's effort outcome; rather, she gets a signal about the agent's action according to a \emph{variable} information structure designed by a regulator. We consider both the case of a risk-neutral and of a risk-averse agent, focusing mainly on a setting with a limited liability assumption. We provide a clean characterization for implementability of actions and utility profiles by any information structure, which turns out to be simple thresholds on the utilities. We further study naturally constrained information structures in which the signal emitted from any action is either the action itself or some actions nearby. We show that the worst implementable welfare deteriorates gracefully as the information structure becomes noisier. Finally, we show that our clean characterization does not generalize to a larger class of signaling constraints. In fact, even deciding whether a certain action is implementable by some constrained information structure from this class is NP-complete in the general~setting.
\end{abstract}

\section{Introduction}
\paragraph{The principal-agent problem.}
In the classic contract design model in microeconomic theory, a principal delegates a project to an agent.
The agent has a ``production technology'', i.e., can choose among a set of available \emph{actions} (a.k.a.~\emph{effort levels}), with each action imposing a \emph{cost} upon the agent. Each action induces a distribution over possible \emph{outcomes}, which determine the principal's payoff from the project. The conflict of interests that arises from the separation between bearing the cost of the project (the agent) and enjoying its outcome (the principal) is known as \emph{moral~hazard}.

Given this misalignment of interests, the principal must incentivize the agent to invest effort in the project --- this is known as the \emph{principal-agent problem}. An underlying, inherent assumption is that the principal does not directly observe the agent's action; therefore, she cannot simply ``purchase'' the agent's work. In the classic principal-agent model, the principal can observe only some outcome induced by the agent's action. To \emph{implement} an action, i.e., to incentivize the agent to take it, the principal pays the agent according to the outcome: She designs a \emph{contract} specifying the monetary \emph{transfer} to the agent as a function of the outcome (notably, not as a function of the unobserved action). The principal's goal in designing the contract is to maximize her expected utility (payoff minus transfer), where the expectation is over the distribution of outcomes induced by the agent's action. 

The economic literature on contract design is extensive (see, e.g., the classic works of~\cite{mirrlees1971exploration,ross1973economic,holmstrom1979moral,grossman1992analysis,mirrlees1999theory} and the textbooks~\cite{macho2001introduction,brousseau2002economics,bolton2004contract,laffont2009theory}). 
Classic applications of the principal-agent problem include designing insurance~\cite{pauly1968economics,arrow1978uncertainty,vera2003structural,rauchhaus2009principal}, governmental contracting~\cite{coats2002applications,roach2016application,mcafee2019incentives}, delegating tasks within an organization~\cite{tommasi2007centralization}, public budgeting~\cite{skok1980budgetary,smith1996principals}, and incentivizing agents to work as a team~\cite{holmstrom1982moral}.
There are also modern, often computerized applications, including crowdsourcing~\cite{HoSV16} and acquiring other online service, incentivizing data sources in statistical learning~\cite{CaiDP15}, designing pay-for-performance healthcare and payments for ecosystem services~\cite{bastani2016analysis,LiAL21}, smart contracts in blockchain~\cite{cong2019blockchain}, 
and device-to-device communication~\cite{hu2019multi}. Such applications have led to a growing interest in contract design from the computational perspective~\cite[see, e.g.,][]{BabaioffFNW12,DuttingRT19,dutting2021complexity} and learning~\cite{cohen2022learning,bacchiocchi2023learning}.

\subsection{Information in the Principal-Agent Problem}
The incomplete information of the principal regarding the agent's action is addressed by the vast contract design literature. With full information, the principal can easily implement the welfare-maximizing action by discriminating the different actions via payments, thus extracting the full welfare as her utility --- the first-best solution. With unknown action, there is tension between economic value and informativeness --- an economically-efficient action can be hard to infer from its outcome, and thus costly to implement, and vice versa. The welfare loss relative to the first-best solution, as well as the principal's utility loss from delegating the project, are both due to this information gap.

Note that in the classic principal-agent model described above, the outcomes have a dual role --- they simultaneously specify the principal's payoff and provide the principal with information about the agent's action. In other words, the production technology mapping actions to outcomes also serves as a ``monitoring technology'' of the principal over the agent. This dual role makes it difficult to perform a comprehensive study of the power of different monitoring technologies for overcoming the information gap.
 
\paragraph{Outcomes versus signals.} 
In many principal-agent interactions, it is natural to assume a separation between the production and monitoring technologies. That is, the principal does not directly observe the outcome of the agent's action; instead, the principal observes a noisy \emph{signal} containing information about the action. We consider an \emph{information structure} specifying the information sent to the principal given every agent's action. In contrast to the classic dual role of outcomes, the information structure in our setting 
is ``responsible'' only for informing the principal, not for determining her income.

Such consideration of separation of the roles between signals and outcomes has been suggested previously, e.g.,~by ~\cite{mirrlees1976optimal,grossman1992analysis}. However, the literature has focused largely on understanding how noisy signals affect the principal's ability to extract utility from the contract. This direction in research established the \emph{informativeness principle}, by which more informative signaling decreases the payment required to implement an action~\cite{holmstrom1979moral,shavell1979risk,gjesdal1982information,grossman1983implicit}. However, as stated by~\citet{salanie2005economics}, ``\emph{the relation `being more informative than' is only a very partial order over the set of possible technologies}''. The aforementioned dual role of the signals in the principal-agent model merges the production and monitoring technologies as a randomized mapping from actions to outcomes. In this paper, we single out the monitoring part, captured by the information structure. The contribution of this paper is to study \emph{all} possible information structures, and how they can affect the optimal contract's welfare and utility.

\subsection{Designing the Information}
\label{sub:info}
We seek to understand the power of designing the information structure by a \emph{regulator} (a.k.a.~\emph{social planner}); for examples of social planners, see Subsection~\ref{sub:apply}. Note that after the regulator fixes the information structure, we are back in a classic principal-agent setting, in which the principal designs an optimal contract and the agent reacts by choosing an (unobservable) action. We say that an action $i$ is \emph{implementable under an information structure} if after the information structure is set in place by the regulator, there exists an optimal contract by the principal that incentivizes the agent to take action $i$. In other words, an equilibrium of the resultant principal-agent problem settles at the agent taking action $i$. We are primarily interested in understanding both agents' utilities and the implemented action at this equilibrium. This leads us to study the following basic research question, which surprisingly has not been thoroughly examined previously.

\begin{question*}
Which utility profiles for the principal and agent (and which welfare obtained from their sum) are implementable under \emph{some} information structure, assuming that the principal chooses an optimal contract?
\end{question*}

The above question views the principal-agent problem from the perspective of the regulator, who can design the information structure to determine what, and how much the principal can observe about the agent's action. We do not specify the regulator's objective; instead, we focus on a more general characterization question to understand all possible outcomes of the principal-agent interactions across \emph{all possible information structures}. As a corollary of our result, a regulator with any natural objective function can maximize the objective over the domain of the implementable~outcomes.

\paragraph{Information design for moral hazard versus adverse selection.} The work that is most relevant to ours in spirit is perhaps the celebrated recent work of~\citet{bergemann2015limits}, who study a very similar question --- i.e., implementability of utility profiles by an information structure --- but in a different setup of \emph{monopoly pricing}. Their work is an example of applying the flourishing field of \emph{information design} to one of the most basic adverse selection problems, namely price discrimination. We wish to achieve similar insights to those arising from their work, but for information design under \emph{moral hazard} (rather than monopoly pricing under \emph{adverse selection}). Interestingly, despite being obtained from quite different techniques, the conceptual messages in our principal-agent problem versus monopoly pricing turn out to be similar --- information design is sufficiently powerful to implement any utility pair above some simple thresholds for both participants.\footnote{The participants are seller and consumer in monopoly pricing, and principal and agent in our setting.}
Through tackling the above question in a moral hazard setting, we hope to spur more work on the wider research frontier of information design for moral hazard --- in Section~\ref{sec:con}, we highlight several concrete future research directions at the intersection of information and contract design. For further discussion of~\cite{bergemann2015limits}, see Section~\ref{sec:related}.

\subsection{Applications}
\label{sub:apply}
Our model applies to many contractual settings in which the monitoring technology is designed by a separate entity from the contracting parties. Computational advancements make the problem of designing the information structure particularly timely, as these enable a wide, fine-grained range of monitoring levels~\cite{li2020optimal}.

To showcase its applications, we mention a few motivating examples here. As a first example, consider an employee (agent) working on a long-term project. The outcome of the employee's effort will be observable by the direct manager (principal) only in the future when the project concludes. The exact amount of effort is also unobservable, given the manager's limited ability to follow her employee's every move. A policy fixed by the company (regulator) allows the manager a certain set of supervision abilities (monitoring technology). The company may have different objectives from the direct manager, such as maintaining sufficiently high utility for the employee, maximizing the company's welfare, and protecting the employee's privacy. Moreover, a company may wish to encourage the employees to take a \emph{specific} action that would benefit the company's long-term~revenue.

Motivated by the restrictions in real world, it is of importance to study constrained design in which not all information structures are available to the regulator. One of the arguably most natural constraints is that the signal a direct manager observes about the employee's action cannot be too noisy. For instance, the manager might inherently have an estimation of the employee's effort containing a bounded noise; the company may wish to enforce a monitoring policy reducing this noise even further. We capture such applications by an abstracted model introduced in Section~\ref{sec:noise}.

Secondly, consider an institute (principal) and a candidate (agent). The candidate's action is the level of effort she invests in preparing for a position at the institute, for which the institute compensates her through her salary. The institute gets an indication about the agent's preparedness from a standardized test designed by an external body (social planner), whose objectives are typically different from both the institute and the candidate~\cite[see, e.g.,][]{kokovikhin2016national}.

Thirdly, our framework is also applicable to platforms for services that match two parties (e.g.,~freelancing platform Upwork that matches tasks with workers), on which the monitoring technology is usually determined by a third party --- the platform itself. Notably, in some of these applications, the social planner may be constrained about what types of signals can potentially be sent. This is another possible motivation for studying constrained information structures.

A family of applications involves monitoring via outsourced advanced technology~-~e.g., monitoring of high-tech employees through commercial software run on their computers. The company creating the software (the regulator) has flexibility in its design and commitment power (due to credibility or legislative issues), but it cannot enforce payments. Another example is monitoring a salesperson’s performance by automated analysis of the sales call: ``running speech analytics software consumes server space and power, and [...] has been increasingly outsourced to third parties"~\cite{li2020optimal}. A very recent example is the auditing of a generative learning agent to ensure alignment with the principal -- a challenging task often requiring third-party cutting-edge technology. Finally, our model captures the case of an agent who is also a regulator in the sense that the agent commits to a schedule of self-reporting or inspections.

\subsection{Our Results}
\label{sub:contribution}
Somewhat surprisingly, even though the principal-agent problem with unobserved outcomes has appeared in previous work (see Section~\ref{sec:related} for a detailed discussion of related work), to the best of our knowledge, ours is the first work to study the implementability of utility profiles and expected transfers. We show that the flexibility of choosing the information structure is powerful: Every possible utility profile that exceeds some natural thresholds for the principal and agent can be implemented by some information structure; see Figures~\ref{fig:1} and~\ref{fig:3} for a visualization of the implementable utility profiles for risk-neutral and risk-averse agents, respectively. 

Furthermore, every action the welfare of which exceeds the welfare of the least costly action can be implemented. These conditions are necessary and sufficient; thus, this is a tight characterization of implementable utility profiles and agent's actions. We further show that this characterization holds for both a risk-neutral agent (Section~\ref{sec:risk-neutral}) and a risk-averse agent (Section~\ref{sec:risk-averse}), but the set of implementable utility pairs under risk-aversion becomes significantly richer. The richness of the sets of implementable actions and utility profiles prescribes the power of the regulator with the ability to design the information structure.\footnote{All our results are constructive, and thus naturally admit polynomial-time algorithms if one is interested in the computational complexity. That is, we can explicitly construct the information structure that induces any implementable utility profile and action at~equilibrium.}

In Section~\ref{sec:noise}, we consider a natural constrained version of the problem motivated by inherent limitations on information structures discussed in the previous subsection. We identify the set of signals directly as the set of actions $\brackets*{n}$. For any fixed \emph{noise level} described by an integer $d\geq 0$,  if the agent takes action $a$, the induced signal is assumed to be supported on signals from the set $\brackets*{a-d,a+d}$ (with obvious truncation if $a-d$ or $a+d$ exceeds the boundary). When $d=0$, this captures one extreme case in which the principal observes the agent's action, and the regulator has no flexibility at all. When $d=n-1$, this captures the other extreme case in which all information structures are available for the regulator. In the direct manager-employee application, $d$ may represent a bound on the ability of the direct manager to immediately distinguish different levels of efforts of her employee, without imposing a stricter monitoring policy.

We characterize the set of implementable actions for all intermediate values of $d$. Specifically, an action is implementable if and only if its welfare is weakly above the minimal welfare of every $2d+1$ consequent actions (Theorem~\ref{thm:noise}). This result nicely demonstrates the efficiency loss that can be caused by the inaccuracy of the principal's information. When $d=0$, only the most efficient outcome is implementable. However, such an outcome might be undesirable by the regulator, since the employee's revenue is $0$, and the employee's privacy might be violated. As $d$ increases, more and more inefficient outcomes might arise, yet this phenomenon occurs~gradually. We further prove that a slightly enriched class of  constraints on the information structures -- namely, upper-bounding the allowed probability of mapping an action to some signal -- would make it NP-complete to decide whether a certain action is implementable (Proposition~\ref{pro:hard}).

\subsection{Paper Organization}
\label{sub:org}
In Section~\ref{sec:pre}, we provide the formal model for our paper. In Sections~\ref{sec:risk-neutral} and~\ref{sec:risk-averse}, we characterize the sets of implementable actions and expected utility profiles when the agent is risk-neutral and risk-averse, respectively. In Appendix~\ref{ap:further}, we consider the implications of dropping the limited liability assumption on the results of these sections. In Section~\ref{sec:noise}, we study the implementable actions under a family of information structures representing increasingly noisy observations of the outcomes, and provide a computational hardness result for a more general family of information structures (Subsection~\ref{sub:negative}). Section~\ref{sec:related} discusses related papers, and Section~\ref{sec:con}~concludes.

\section{Model}
\label{sec:pre}
The \emph{principal-agent problem} involves two strategic players --- the \emph{principal} and the \emph{agent}. The principal's design problem is given by a tuple $\parentheses*{n,r,c,k,I}$ with the following interpretation:
\begin{itemize}
    \item $\brackets*{n}:=\braces*{0,1,\ldots,n}$ is the set of possible actions of the agent, where $n\in \mathbb{N}$.\footnote{Our results can be easily generalized to an agent with a continuum of possible actions.}
    \item $c=\parentheses*{c_0,c_1,...,c_n}\in\mathbb{R}_{> 0}^{n+1}$ is the vector of the agent's \emph{costs} for taking the corresponding actions. Without loss of generality, we assume that $0\leq c_1 \leq ...\leq c_n$.
    \item $r=\parentheses*{r_0,r_1,...,r_n}\in\mathbb{R}_{\geq 0}^{n+1}$ is the principal's vector of \emph{incomes} (in expectation over the outcomes) induced by the corresponding actions. We distinguish between the principal's income induced by the action only, and her overall \emph{utility} that also considers her transfer to the agent.
    \item The action $0$ is a special \emph{default} action with $c_0=0$. Namely, the agent always has the option to exert no effort. We do not assume anything on $r_0$, besides that it is non-negative.
    \item As is natural, we assume that for $i\neq i'$, it holds that $\parentheses*{c_i,r_i}\neq \parentheses*{c_{i'},r_{i'}}$.
    \item $k\in \mathbb{N}$ is the number of possible \emph{signals} the principal can observe.
    \item A row-stochastic matrix $I=\parentheses*{I_{i,j}}\in \mathbb{R}^{n\times k}$ is the \emph{information structure}.\footnote{A \emph{row-stochastic matrix} is a real matrix with non-negative entries s.t.~each row sums up to $1$.} The $i$-th row $I_i$ of $I$ ($1\leq i\leq n$) specifies the distribution over the $k$ signals when the agent's action is $i$.\footnote{We remark that an alternative model could allow the regulator to randomize over information structures so that both the principal and the agent observe the realized information structure before the principal chooses a contract. The results in this paper do not change significantly in such a model. In fact, such a model convexifies the problem and potentially simplifies some arguments in the proofs. However, we see an added value in proving the stronger results on the existence of an information structure with certain properties rather than a distribution over those.} Note that when the agent chooses to take the default action $0$ --- it is revealed to the principal.
\end{itemize}

A \emph{contract} is a vector $t=\parentheses*{t^1,t^2,...,t^k}$, where $t^j$ is the principal's transfer to the agent upon observing the signal $j$. In the main Sections~\ref{sec:risk-neutral}-\ref{sec:noise}, we assume \emph{limited liability}~\cite{sappington1983limited,innes1990limited}: $t^j \geq 0$ for every $1\le j\le k$ --- i.e., the principal cannot charge the agent; however, we show in Appendix~\ref{ap:further} that similar results hold when the principal may charge the agent. The induced expected transfer at action $a\in \braces*{1,\ldots,n}$ is denoted by $t_a:=\mathbb{E}_{j\sim I_a}\brackets*{t^j}$. We denote $t_0:=0$, as the transfer for action $0$ is always $0$.

The \emph{principal's utility} is $u^P_a=r_a-t_a$ --- i.e., her income for the agent's action minus the expected monetary transfer. Similarly, the agent's utility depends both on her action and the principal's transfer. When the agent is \emph{risk-neutral} (Sections~\ref{sec:risk-neutral} and~\ref{sec:noise}), her expected utility is linear. Specifically, for an action $a$ and a contract $t$, the agent's utility is $u^A_a=t_a-c_a$. When the agent is \emph{risk-averse} (Section~\ref{sec:risk-averse}), we take a concave von Neumann-Morgenstern function $v:\mathbb{R}_{\geq 0}\to \mathbb{R}_{\geq 0}$ to capture the part of agent's utility originating from the monetary transfer. The total agent's utility at action $a\in \brackets*{n}$ is $u^A_a:=\mathbb{E}_{j\sim I_a} \brackets*{v\parentheses*{t^j}}-c_a$.

The principal-agent interaction is a Stackelberg game~\cite{stackelberg1952theory} in which the principal chooses a contract, and thereafter the agent chooses an action. The agent picks an action maximizing his expected utility; the principal selects a contract maximizing her expected utility assuming that the agent will pick his best action given the contract. As is standard, we consider \emph{pure} subgame perfect equilibria in which the agent brakes ties in the principal's favor. Note that the utility profile at equilibrium is uniquely determined.

\paragraph{Implementability.} We say that an information structure $I$ \emph{implements the utility profile $\parentheses*{x,y}\in\mathbb{R}_{\geq 0}\times \mathbb{R}_{\geq 0}$} if there exists an equilibrium outcome of the Stackelberg game $\parentheses*{n,r,c,k,I}$ for which $u^P_a=x$ and $u^A_a=y$. Similarly, we say that $I$ \emph{implements the action-transfer pair $\parentheses*{a,s}\in \brackets*{n}\times \mathbb{R}_{\geq 0}$} if the chosen action at the equilibrium is $a$ and the expected transfer is $s$. Call $\parentheses*{x,y}\in \mathbb{R}_{\geq 0}\times \mathbb{R}_{\geq 0}$ (or $\parentheses*{a,s}\in \brackets*{n}\times \mathbb{R}_{\geq 0}$) \emph{implementable} if there exists an information structure $I$ that implements the utility profile $\parentheses*{x,y}$ (or the action $a$ and the transfer $s$). \emph{The goal of this paper is to characterize the set of implementable utility profiles $\parentheses*{x,y}$ and action-transfer pairs~$\parentheses*{a,s}$}.

\section{Risk-Neutral Agent}
\label{sec:risk-neutral}
In this section, we consider a risk-neutral agent. That is, the agent's (expected) utility at action $a$ given a contract $t$ is $u^A_a=t_a-c_a$. The \emph{social welfare} at action $a$ is $w_a=u^P_a+u^A_a=r_a-c_a$. Note that $w_0=r_0$.

The set of possible utility profiles is:
\begin{align*}
F:=\braces*{\parentheses*{r_i-s,-c_i+s}:\;i\in \brackets*{n},\; s\in \mathbb{R}_{\geq 0}}\cup \braces*{\parentheses*{r_0,0}}
\end{align*}
(see Figure~\ref{fig:1}). Note that with a risk-neutral agent, the implementability of an action-transfer pair $\parentheses*{a,s}$ directly implies the implementability of the utility pair $\parentheses*{r_a-s,s-c_a}$.

\begin{figure}[h]
    \centering
    \scalebox{0.8}{
    \begin{tikzpicture}
    \draw (-1,0)--(2,0);
    \draw[blue] (2,0)--(7,0);
    \draw[->] (7,0)--(8,0);
    \node[right] at (8,0) {$u^P$};
    
    \draw[->] (0,-3.5)--(0,6);
    \node[above] at (0,6) {$u^A$};
    
    \draw (4,-2)--(-1,3);
    \draw (5,-4)--(-1,2);
    \draw (7,-3.5)--(-1,4.5);
    \draw (8,-3)--(-1,6);
    
    \draw[line width=1mm] (3.5,0)--(2,1.5);
    
    \draw[line width=1mm] (5,0)--(2,3);
    
    \filldraw (4,-2) circle(0.1);
    \node[below] at (4.2,-2) {$\parentheses*{r_1,-c_1}$}; 
    
    \filldraw (8,-3) circle(0.1);
    \node[below] at (8,-3) {$\parentheses*{r_i,-c_i}$};
    
    \filldraw (2,0) circle(0.1);
    \node[below] at (2,-0.1) {$w_1$};
    
    \node[below] at (5,-0.1) {$w_i$};
    
    \node[right] at (2,3) {$(w_1,w_i-w_1)$};

    \draw[blue] (2,0)--(2,6);
    
    \end{tikzpicture}}
    \caption{The union of all $n$ lines is the set $F$. In this figure, $w_1\geq 0$, $r_0=0$, and the bold parts of the lines are the implementable utility profiles. The index $i$ denotes the social welfare-maximizing~action.}
    \label{fig:1}
\end{figure}
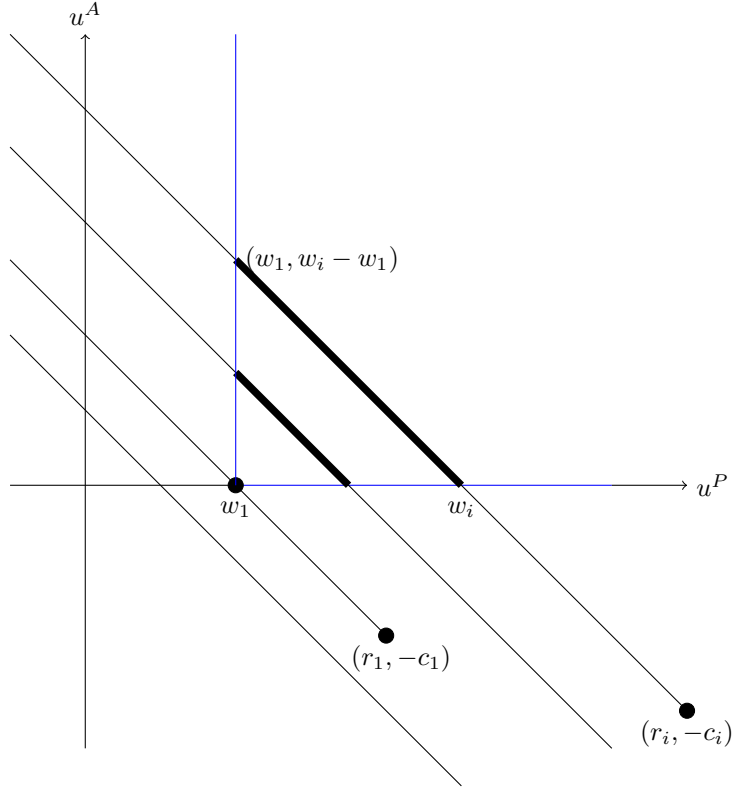

We shall characterize the set of implementable utility profiles $\parentheses*{x,y}$. Immediate bounds on the utilities are $x\geq w_1$, $x\geq w_0$, and $y\geq 0$. Indeed, the principal can get the utility of $w_1=r_1-c_1$ by constantly transferring to the agent $c_1$ (i.e., the contract $\forall j:\;\;t^j= c_1$, which induces action $1$ and a utility of $w_1$ for the principal). The inequalities $x\geq w_0$ and $y\geq 0$ follow from the individual rationality of the principal and the agent, respectively. We shall prove that every utility profile that satisfies these simple constraints is implementable.

The corresponding information structure that implements these utility profiles turns out to have binary signals. Moreover, the implementation is universal in the following sense: A specific information structure parameterized by two values that only depend on the cost vectors and the required expected transfer fits for implementing \emph{all} the implementable action-transfer pairs. This scheme pools together all the actions with cost not smaller than that of the implemented action and a carefully chosen fraction of the probability mass on the less costly actions, summarized together into a \emph{high} signal; and it sends a \emph{low} signal for the remaining probability mass on the less costly actions. The information structure eliminates from consideration all the actions but $0$, $1$, and the action we wish to implement.\footnote{This technique is known as \emph{delegation}; the designer of the information structure reduces the number of inducible actions to just three, and the design of the contract that would make the agent choose a specific action from the three is delegated to the principal. Delegation in contract design appears, e.g., in~\cite{armstrong2010model}.} The optimal contract pays $0$ when the signal is low and an appropriate $s$ when the signal is high. See further intuition about the result in Example \ref{ex:agent-max}.
\begin{theorem}\label{thm:utility_neutral}
A utility profile $\parentheses*{x,y}\in F$ is implementable if and only if $x\geq \max\braces*{w_0
,w_1}$ and $y\geq 0$. Moreover, an action $i\in \brackets*{n}$ is implementable (for some expected transfer) if and only if $w_a\geq \max\braces*{w_0,w_1}$. If action $a\in \brackets*{n}$ is implementable, then the set of implementable action-transfer pairs is $\braces*{\parentheses*{a,s}: s\in \brackets*{c_a, r_a-\max\braces*{w_0,w_1}}}$, and using just two signals suffices to implement any such pair.
\end{theorem}

We shall now discuss two special cases of Theorem~\ref{thm:utility_neutral}.

\begin{example}\label{ex:agent-max}
Consider a social planner with perfectly aligned incentives with the agent; i.e., the social planner aims to maximize the agent's utility. It follows from Theorem~\ref{thm:utility_neutral} that the agent-best action that can be implemented is the welfare-maximizing action $a^*$;\footnote{If there are several welfare-maximizing actions --- implementing them is equally good for the agent.} moreover, the maximal agent's utility is $w_{a^*}-\max\braces*{w_0,w_1}$. From Theorem~\ref{thm:utility_neutral} proof, it follows that the following binary-signal information structure $I$ implements $a^*$: It sends the \emph{high} signal with probability $1$ when the agent takes either the action $a^{*}$ or a more costly action; otherwise, it sends the high signal with probability $1-\parentheses*{c_{a^*}-c_1}/\parentheses*{r_{a^*}-\max\braces*{w_0,w_1}}$ (and with the remaining probability, it uses the \emph{low}~signal).

To gain intuition about the meaning of these expressions, let us analyze the case $w_{a*}>w_1\geq w_0$. Let us consider the two actions $1,a^*$ only and ignore the remaining actions; later, we shall justify this step. Consider the following class of information structures that is parameterized by a single number $\alpha\in \brackets*{0,1}$:

\begin{table}[h]
\begin{center}
\begin{tabular}{ccc}
Signal:& Low                             & High                          \\ \cline{2-3} 
\multicolumn{1}{c|}{1}     & \multicolumn{1}{c|}{$1-\alpha$} & \multicolumn{1}{c|}{$\alpha$} \\ \cline{2-3} 
\multicolumn{1}{c|}{$a^*$} & \multicolumn{1}{c|}{0}          & \multicolumn{1}{c|}{1}        \\ \cline{2-3} 
\end{tabular}
\end{center}
\end{table}

For $\alpha = 0$, the principal can extract the entire welfare by paying $s=c_{a^*}$ at the high signal. The same is true for all $\alpha\leq \frac{c_1}{c_{a^*}}$. However, when $\alpha$ exceeds the threshold $\frac{c_1}{c_{a^*}}$, the principal can no longer extract the entire welfare, as the agent will prefer to take action $1$ under that contract. The optimal contract in such a case would be to increase the transfer $s$ for the high signal to $s=\frac{c_{a^*}-c_1}{1-\alpha}$; the agent would gain a positive utility of $\frac{c_{a^*}-c_1}{1-\alpha}-c_a^*$. When $\alpha$ reaches another threshold of $\alpha^*=1-\frac{c_{a^*}-c_1}{r_{a^*}-w_1}$, and $s$ reaches $r_{a^*}-w_1$, the principal becomes indifferent between using the above contract to incentivize action $a^*$ and just making a constant transfer of $c_1$ to incentivize action $1$. This threshold $\alpha^{*}$ gives the agent the optimal expected utility of $\parentheses*{r_{a^*}-w_1}-c_a^*=w_{a^*}-w_1$; see Figure~\ref{fig:4} for illustration.

\begin{figure}[h]
\caption{The expected principal's (blue) and agent's (red) utilities as a function of $\alpha$.}
    \label{fig:4}
    \centering
    \begin{tikzpicture}[scale=3]
    \draw[->] (-0.05,0) -- (1.05,0);
    \draw[->] (0,-0.05) -- (0,1.05);
    \draw[red] (0,0.01)--(0.3,0.01);
    \draw[blue] (0,0.95)--(0.3,0.95);
    \draw[red] (0.7,0.01)--(1,0.01);
    \draw[blue] (0.7,0.45)--(1,0.45);
    \draw[red] plot [domain=0.3:0.7,samples=100] (\x,{0.2625/(1-(\x))-0.2625/0.7});
    \draw[blue] plot [domain=0.3:0.7,samples=100] (\x,{0.95-(0.2625/(1-(\x))-0.2625/0.7)});
    \draw (1,-0.02)--(1,0.02);
    \node[below] at (0.3,-0.02) {$\frac{c_1}{c_{a^*}}$};
    \node[below] at (0.7,-0.02) {$\alpha^{*}$};
    \node[left] at (-0.02,0.5) {$w_{a^{*}}-w_1$};
    \node[below] at (1,-0.02) {1};
    \node[right] at (1.05,0) {$\alpha$};
    \node[above] at (0,1.05) {$\text{Utility}$};
    \node[below] at (0.02,0) {0};
    \node[left] at (0,1) {$w_{a^*}$};
    \end{tikzpicture}
\end{figure}

To explain why we focus on actions $1,a^*$ only, note that one can set the distribution of signals to be $\parentheses*{0,1}$ for all actions that are more costly than $a^*$. In such a case, the agent will always prefer action $a^*$ over the more costly actions. We can also set the distribution of signals to be $\parentheses*{1-\alpha,\alpha}$ for all actions that are less costly than $a^*$. In such a case, the agent will always prefer action $1$ over actions that are less costly than $a^*$.
\end{example}

\begin{example}
\label{ex:strategic}
Consider a strategic social planner with utility $p_a$ that depends on the agent's action $a$. The planner aims to select an information structure maximizing her utility for the implemented action. It follows from Theorem~\ref{thm:utility_neutral} that the planner should simply induce $a$ that maximizes $p_a$ s.t.~$w_a\geq\max\braces*{w_0,w_1}$.
\end{example}

\begin{proof}[Proof of Theorem~\ref{thm:utility_neutral}]
The necessity of $x\geq \max\braces*{w_0,w_1}$ and $y\geq 0$ was discussed before the theorem. We shall prove the desired necessary and sufficient conditions on implementability of the pairs $\parentheses*{a,s}$. It would immediately imply that $\parentheses*{x,y}\in F$ is implementable for the action $a$ and a transfer $s=r_a-x$ if and only if $ r_a-x\leq r_a-\max\braces*{w_0,w_1}$ and $r_a-\parentheses*{w_a-y}\geq c_a$ --- equivalently, if and only if $x\geq \max\braces*{w_0,w_1}$ and $y\geq 0$, as required.

Given the pair $\parentheses*{a,s}\in \brackets*{n}\times \mathbb{R}_{> 0}$, define $p:=1-\frac{c_a-c_1}{s}$.\footnote{If $s=0$, then trivially $\parentheses*{a,s}$ cannot be implemented --- the agent will better off taking action $0$, as $c_a>0$. We shall assume, therefore, that $s>0$.} Note that $p\in \brackets*{0,1}$, as 
\begin{align*}
  s\geq c_a >c_a-c_1 \Rightarrow 1-\frac{c_a-c_1}{s}>0.  
\end{align*}
Consider the following binary-signal ($k=2$) information structure:  
For the action $i=a$ and the actions $i\in \brackets*{n}$ s.t.~$c_i>c_a$ --- set  $I_i:=\parentheses*{1,0}$; for all the remaining actions $i$ --- take $I_i:=\parentheses*{p,1-p}$. We argue that the information structure $I$ implements the pair $\parentheses*{a,s}$.

Note that for the selected $I$, the actions $i\notin\braces*{0,1,a}$ are suboptimal for the agent under any contract. Indeed, any $i$ that is more costly for the agent than $a$ yields the same transfers as $a$ does --- and hence is inferior to $a$; any other $i$ is more costly than action $1$, but as it yields the same transfers as action $1$ does --- it is inferior. Henceforth, actions $i\notin \braces*{0,1,a}$ can be ignored.

We claim that the contract $t_*=\parentheses*{t^1_*,t^2_*}=\parentheses*{s,0}$ is optimal, and it induces action $a$. First, note that for $t_*$, the agent's utility upon taking actions $a$ or $1$ is $s-c_a\geq 0$, and upon taking action $0$ is $0$. By our tiebreak rule, she chooses $a$; therefore, the principal's utility under $t_*$ is $r_a-s$. It remains to show that $r_a-s$ is the maximal utility that the principal can ensure.

Indeed, any contract yielding action $1$ must have an expected transfer $t_1\geq c_1+w_0$ (otherwise, the agent will prefer action $0$). For such contracts, the principal's utility is $r_1-t_1\leq r_1-c_1$. By assumption, $r_a-s\geq r_1-c_1=w_1$; hence, those contracts are not better than $t_*$. Moreover, any contract yielding the action $0$ results in the principal's utility of $w_0$, which is not more than $r_a-s$. Furthermore, a contract $t=\parentheses*{t^1,t^2}$ yielding the action $a$ must satisfy $t^1 - c_a \geq \min\braces*{t^1, pt^1+\parentheses*{1-p}t^2}-c_1$ (otherwise --- the agent will prefer action $1$). Therefore, 
\begin{align*}
   t^1 - c_a \geq pt^1-c_1 \Rightarrow t^1\geq \frac{c_a-c_1}{1-p}=s. 
\end{align*}
Thus, the principal's utility is at most $r_a-s$, as desired.
\end{proof}
Notably, the information structure that implements the desired utility profile in the above proof satisfies a \emph{monotonicity} property --- i.e.,  more costly actions always induce ``higher'' signals in the first-order stochastic-dominance sense. More accurately, the signaling scheme uses a ``high'' and a ``low'' signal in our construction; a more costly action always induces a higher probability of sending the ``high'' signal. Such monotonicity property is sometimes assumed in the classical principal-agent problem, where the stochastic dominance is over the outcomes (see, e.g.,~\cite{rogerson1985first,sinclair1994first,alvi1997first}). As a corollary of Theorem~\ref{thm:utility_neutral} --- whatever can be implemented by an arbitrary information structure, can also be implemented by a monotonic information structure.

\section{Risk-Averse Agent}\label{sec:risk-averse}
In this section, we consider an agent with a concave von Neumann-Morgenstern utility $v:\mathbb{R}_{\geq 0}\to \mathbb{R}_{\geq 0}$. 
We assume that $v\parentheses*{0}=0$, $v$ is strictly increasing, and $\lim_{z\to \infty}\frac{v\parentheses*{z}}{z}=0$. Recall that the agent's utility at action $i\in \brackets*{n}$ is given by $u^A_i=\mathbb{E}_{j\sim I_i} \brackets*{v\parentheses*{t^j}}-c_i$.

We shall first specify the set of possible utility profiles. If the agent chooses action $i\in \brackets*{n}\setminus\braces*{0}$ and receives a deterministic transfer of $z$, the utility profile is $\parentheses*{r_i-z,v\parentheses*{z}-c_i}$. If the transfer is random with expectation~$z$, the utility profile may change to $\parentheses*{r_i-z,y}$, for some $y\leq v\parentheses*{z}-c_i$.~Denote:
\begin{align*}
    F_i:=\braces*{\parentheses*{x,y}:\; x=r_i-z,\; y\leq v\parentheses*{z}-c_i \emph{ for some } z\in \mathbb{R}_{\geq 0}}=\braces*{\parentheses*{x,y}:\; x\leq r_i,\; y\leq v\parentheses*{r_i-x}-c_i}
\end{align*}
(see Figure~\ref{fig:2} for visualization of $v$ and $F_i$).

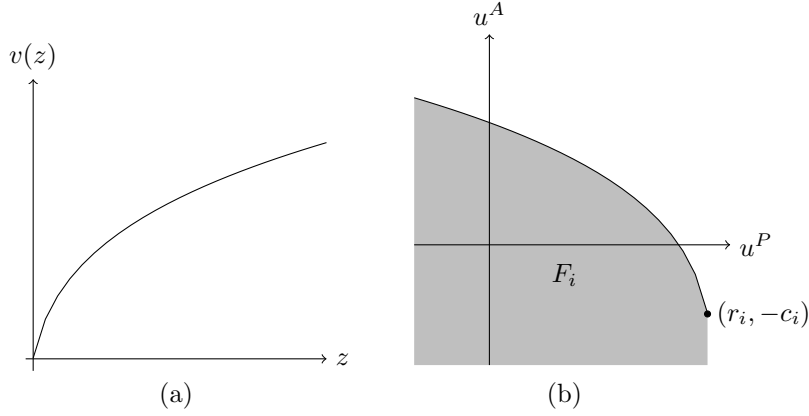
\begin{figure}[h]
    \centering
    \begin{tikzpicture}[xscale=2, yscale=4]
    
    \draw[->] (0,0.47)--(2,0.47);
    \node[right] at (2,0.47) {$z$};
    
    \draw[->] (0.05,0.43)--(0.05,1.4);
    \node[above] at (0.05,1.4) {$v(z)$};
    \draw [domain=0.05:2, variable=\x]
      plot ({\x}, {sqrt{sqrt{\x}}});
    \node at (1,0.35) {(a)};
    \end{tikzpicture}
    \hspace{5mm}
    \begin{tikzpicture}[xscale=2,yscale=4]
        \fill [gray!50!white, domain=0.05:2, variable=\x]
      (-0.05,0.3)--
      plot ({-\x}, {sqrt{sqrt{\x}}})
      -- (-2,0.3) -- cycle;
        
        \draw [domain=0.05:2, variable=\x]
      plot ({-\x}, {sqrt{sqrt{\x}}});
      \draw[->] (-2,0.7)--(0.1,0.7);
        \node[right] at (0.1,0.7) {$u^P$};
        
        \draw[->] (-1.5,0.3)--(-1.5,1.4);
        \node[above] at (-1.5,1.4) {$u^A$};
        
      \filldraw (-0.05,0.47) ellipse (0.02 and 0.01);
      \node[right] at (-0.05,0.47) {$\parentheses*{r_i,-c_i}$};
      \node at (-1,0.6) {$F_i$};
      \node at (-1,0.2) {(b)};
    \end{tikzpicture}
    \caption{The function $v$ appears in (a). The set $F_i$ appears in (b).}
    \label{fig:2}
\end{figure}

Note that $F_i$ is a superset of the possible utility profiles that can be induced from the action $i$. The exact set of possible utilities depends on the information structure $I$. For our purposes, the sets $\braces*{F_i}_{1\leq i\leq n}$ would suffice to characterize the implementable utility profiles. We denote by $F:=\cup_{1\leq i\leq n} F_i \cup \braces*{\parentheses*{r_0,0}}$ the superset of the possible utility profiles; see Figure \ref{fig:3}.

Again, our goal is to characterize the set of the implementable utility profiles $\parentheses*{x,y}\in F$. Trivial bounds on the utilities are $x\geq r_1-v^{-1}\parentheses*{c_1}$, $x\geq w_0$, and $y\geq 0$. Indeed, the principal can achieve utility of $r_1-v^{-1}\parentheses*{c_1}$ by always transferring the agent $v^{-1}\parentheses*{c_1}$ (i.e., the contract $t^j\equiv v^{-1}\parentheses*{c_1}$ for all $1\leq j\leq k$). The inequalities $x\geq w_0$ and $y\geq 0$ follow from individual rationality, similarly to the risk-neutral agent case. We shall show that every utility profile satisfying these inequalities is implementable; see Figure~\ref{fig:3}.

As in the risk-neutral case, there exists a \emph{universal binary-signal information structure} implementing all the implementable utility profiles, parameterized by two values derived from the cost vectors, $v$, and the desired utility profile; this information structure satisfies the same monotonicity property. However, the signals themselves take a more intricate form here. Designing the implementing information structure turns out to be rather subtle. We carefully choose the ``high'' and the ``low'' signals, using the concavity of $v$, to ensure that the principal's optimal contract yields risky monetary transfers for the agent. Specifically, with a positive probability, our risk-averse agent gets zero monetary transfer; note that by the concavity of $v$, the expected agent's income from the transfers is smaller than the expected principal's transfer. We delicately choose the distributions over signals in the information structure to get such a second (nonzero) monetary transfer in the optimal contract that the agent will take the action we wish to implement.

\begin{figure}[h]
    \centering
    \begin{tikzpicture}[xscale=3.5,yscale=7]
        
      \fill [blue, domain=0.95:1.19, variable=\x]
      (-0.5,0.65)--
      plot ({-\x}, {sqrt{sqrt{\x}}-0.1})
      --(-1.18,0.65) -- cycle;
      
      \fill [blue, domain=0.8:1.95, variable=\x]
      (-1.18,0.65)--
      plot ({-\x+1}, {sqrt{sqrt{\x}}-0.3})--cycle;

      \draw [dashed, domain=0.05:0.8, variable=\x]
      plot ({-\x-1}, {sqrt{sqrt{\x}}});
        
      \draw [ domain=0.95:1.8, variable=\x]
      plot ({-\x}, {sqrt{sqrt{\x}}-0.1});
      
      \draw [dashed, domain=0.05:0.95, variable=\x]
      plot ({-\x}, {sqrt{sqrt{\x}}-0.1});

      \draw [domain=0.05:1.95, variable=\x]
      plot ({-\x+1}, {sqrt{sqrt{\x}}-0.3});
      
      \draw [dashed, domain=0.05:2.8, variable=\x]
      plot ({-\x+1}, {sqrt{sqrt{\x}}-0.3});
      
      \draw[->] (-1.8,0.65)--(1,0.65);
      
      \draw[->] (-1.7,0.2)--(-1.7,1.1);
      
      \filldraw (-1.18,0.65) ellipse (0.02 and 0.01);
      
      \filldraw (-1.05,0.47) ellipse (0.02 and 0.01);
      
      \filldraw (-0.05,0.37) ellipse (0.02 and 0.01);
      
      \filldraw (0.95,0.17) ellipse (0.02 and 0.01);

      \node[below] at (-1.05,0.47) {$\parentheses*{r_1,-c_1}$};
      
      \node[below] at (-0.05,0.37) {$\parentheses*{r_i,-c_i}$};
      
      \node[below] at (-1.4,0.65) {$r_1-v^{-1}\parentheses*{c_1}$};
      
      \node[right] at (1,0.65) {$u^P$};
      
      \node[above] at (-1.7,1.1) {$u^A$};
        
    \end{tikzpicture}
    \caption{The area below the solid line is the set $F$. The blue area is the set of implementable utility profiles.}
    \label{fig:3}
\end{figure}
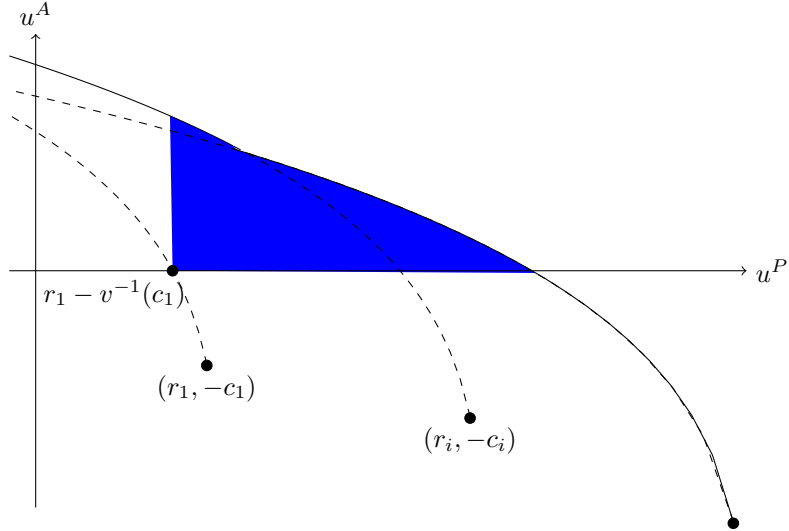

\begin{theorem}
\label{thm:utility_averse}
A utility profile $\parentheses*{x,y}\in F$ is implementable if and only if $x\geq \max\braces*{r_1-v^{-1}\parentheses*{c_1},w_0}$ and $y\geq 0$. Moreover, an action $a\in \brackets*{n}$ is implementable if and only if $r_a-v^{-1}\parentheses*{c_a}\geq \max\braces*{r_1-v^{-1}\parentheses*{c_1},w_0}$. In this case, $\parentheses*{a,s}$ is implementable if and only if $s\in\brackets*{v^{-1}\parentheses*{c_a}, r_a - \max\braces*{r_1-v^{-1}\parentheses*{c_1},w_0}}$, and using just two signals suffices to implement any such pair $\parentheses*{a,s}$.
\end{theorem}
\begin{proof}
The necessity of $x\geq \max\braces*{r_1-v^{-1}\parentheses*{c_1},w_0}$ and $y\geq 0$ has been discussed before the theorem. Now we shall prove the sufficiency. We shall only consider the case $w_1\geq w_0$. When $w_1< w_0$, the same arguments hold, with action $0$ playing the role of action $1$ throughout the entire~analysis.

Let $\parentheses*{x,y}\in F_a$ be a utility profile that satisfies the inequalities. Namely,  $\max\{ r_1-v^{-1}\parentheses*{c_1},w_0\} \leq x \leq r_a$ and $0\leq y \leq v\parentheses*{r_a-x}-c_a$. We shall show that $\parentheses*{x,y}$ can be implemented by a binary-signal information structure $I$ that is given by $I_i=\parentheses*{p,1-p}$ for $i\neq a$ with $c_i\leq c_a$ (the \emph{low signal}), and $I_i=\parentheses*{q,1-q}$ for $i=a$ and any action $i\in \brackets*{n}$ s.t.~$c_i>c_a$ (the \emph{high signal}), where $p$ and $q$ will be defined soon in Equation~\eqref{eq:pq}.

Denote by $d_0:=v'_+\parentheses*{0}$ the right-hand derivative of $v$ at $0$. Consider the function $\frac{v\parentheses*{z}}{z}$. The concavity of $v$ implies that $\frac{v\parentheses*{z}}{z}$ is monotonically decreasing. Note that $\lim_{z\to 0+} \frac{v\parentheses*{z}}{z}=d_0$ by L'Hopital's rule and $\lim_{z\to \infty} \frac{v\parentheses*{z}}{z}=0$ by assumption. By concavity of $v$: $d_0 \cdot \parentheses*{r_a-x}\geq v\parentheses*{r_a-x}$. Moreover, by the assumption on $x,y$: $y \leq v\parentheses*{r_a-x}-c_a$. Combining the last two conditions yields: $0\leq \frac{y+c_a}{r_a-x}\leq d_0$. By the intermediate value theorem and the monotonicity of $\frac{v\parentheses*{z}}{z}$, there exists a unique solution $z=z_{*}\in \mathbb{R}_{\geq 0}$ for $\frac{v\parentheses*{z}}{z}=\frac{y+c_a}{r_a-x}$.\footnote{Here we define $\frac{v\parentheses*{0}}{0}:=d_0$.}

We set:
\begin{align}\label{eq:pq}
    p:=\frac{r_a-x}{z_{*}}\frac{y+c_1}{y+c_a} \text{ and } q:=\frac{r_a-x}{z_{*}}.
\end{align}
The above values were carefully chosen to ensure that under a suitable binary-signal information structure inducing distributions $\parentheses*{p,1-p}$ and $\parentheses*{q,1-q}$ over the signals, the optimal contract --- that will be defined using the intermediate value theorem on $\frac{v\parentheses*{z}}{z}$ --- would yield exactly the utility profile $\parentheses*{x,y}$.

We first verify that $p,q\in \brackets*{0,1}$. The inequalities $p,q\geq 0$ and $p\leq q$ are trivial. It remains to show that $q\leq 1$. Assume by way of contradiction that $r_a-x>z_{*}$. 
We get that:
\begin{align*}
 \frac{v\parentheses*{z_{*}}}{z_{*}}>\frac{v\parentheses*{r_a-x}}{r_a-x}\geq \frac{y+c_a}{r_a-x},
\end{align*}
where the first inequality follows from the monotonicity of $\frac{v\parentheses*{z}}{z}$, and the second inequality follows from $y \leq v\parentheses*{r_a-x}-c_a$. This contradicts the fact that $\frac{v\parentheses*{z_{*}}}{z_{*}}= \frac{y+c_a}{r_a-x}$.

Similarly to Theorem~\ref{thm:utility_neutral} --- for all contracts, the agent prefers action $1$ over action $i\notin \braces*{0,a}$, since such actions provide her same distribution over transfers as action 1 does, but action 1 is less costly. Therefore, henceforth we can ignore all actions except $0,1$, and $a$. 

For the chosen information structure $I$, consider the contract $t_*=\parentheses*{t^1_*,t^2_*}=\parentheses*{z_{*},0}$. Under this contract, the agent's utilities at actions $1$ and $a$ are, respectively,
\begin{align*}
    &u^A_1\parentheses*{t_*}=pv\parentheses*{z_{*}}-c_1=\parentheses*{r_a-x}\frac{y+c_1}{y+c_a}\frac{y+c_a}{r_a-x}-c_1=y,\\
    &u^A_a\parentheses*{t_*}=qv\parentheses*{z_{*}}-c_a=\parentheses*{r_a-x}\frac{y+c_a}{r_a-x}-c_a=y.
\end{align*}
The agent brakes ties in the principal's favour; hence, he chooses action $a$. This action yields a utility of $u^P_a=r_a-qz_{*}=x$ to the principal. It remains to show that $t_*$ is an optimal contract for the~principal.

Note that among the contracts implementing action $1$, the optimal one for the principal is $t=\parentheses*{v^{-1}\parentheses*{c_1},v^{-1}\parentheses*{c_1}}$ (i.e., a deterministic transfer of $v^{-1}\parentheses*{c_1}$). The principal's utility in such a contract is $r_1-v^{-1}\parentheses*{c_1}\leq x$; hence, the principal would prefer the contract $t_*$. Since $x\geq w_0$, we also know that the principal will prefer the contract $t^*$ over those inducing action $0$. It remains to show that among the contracts that implement the action $a$, the contract $t_*$ is the optimal one. 

Indeed, every contract $t=\parentheses*{t^1,t^2}$ that implements the action $a$ must satisfy $u^A_a\parentheses*{t}\geq u^A_1\parentheses*{t}$ and $u^A_a\parentheses*{t}\geq 0$. We first show that $u^A_a\parentheses*{t}\geq u^A_1\parentheses*{t}$ implies $u^A_a\parentheses*{t}\geq 0$. Indeed:\footnote{We assume $q>p$, as for $q=p \implies c_a=c_1$ the claim is immediate.}
\begin{align}\label{eq:a1}
\begin{aligned}
 &u^A_a\parentheses*{t}\geq u^A_1\parentheses*{t} \Leftrightarrow \\
 &qv\parentheses*{t^1}+\parentheses*{1-q}v\parentheses*{t^2}-c_a\geq pv\parentheses*{t^1}+\parentheses*{1-p}v\parentheses*{t^2}-c_1 \Leftrightarrow \\
 &v\parentheses*{t^1}-v\parentheses*{t^2}\geq \frac{c_a-c_1}{q-p}.
 \end{aligned}
\end{align}
The inequality $u^A_a\parentheses*{t}\geq 0$ now follows from:
\begin{align}\label{eq:a0}
\begin{aligned}
    u^A_a\parentheses*{t} & =qv\parentheses*{t^1}+\parentheses*{1-q}v\parentheses*{t^2}-c_a\geq  \\
    & q\parentheses*{ v\parentheses*{t^2}+\frac{c_a-c_1}{q-p}}+
    \parentheses*{1-q}v\parentheses*{t^2}-c_a=   \\
    & v\parentheses*{t^2}+\frac{p}{q-p}c_a-\frac{q}{p-q}c_1=  \\
    & v\parentheses*{t^2}+\frac{y+c_1}{c_a-c_1}c_a-\frac{y+c_a}{c_a-c_1}c_1= \\
    & v\parentheses*{t^2}+y\geq v\parentheses*{0}+y\geq 0,
\end{aligned}
\end{align}
where the inequality between the first and the second lines follows from~\eqref{eq:a1}, and the equality of the third and the fourth lines are obtained by substitution of the values of $p$ and $q$ from~\eqref{eq:pq}.

If $t^2>0$, then slightly reducing $t^2$ preserves the inequality in both~\eqref{eq:a1} and~\eqref{eq:a0}. This operation results in a contract inducing the same action with a lower expected monetary transfer. Thus, the original contract is suboptimal. We are left now only with contracts of the form $t=\parentheses*{t^1,0}$.

If $u^A_a\parentheses*{t}>u^A_1\parentheses*{t}$, then the inequality in~\eqref{eq:a0} is strict and $u^A_a>0$. Such a contract is suboptimal, as the principal can decrease $t^1$ while still inducing the same action. Therefore, $u^A_a\parentheses*{t}=u^A_1\parentheses*{t}$ and the equation in~\eqref{eq:a1} holds with equality. Namely,
\begin{align}\label{eq:f}
    v\parentheses*{t^1}=\frac{c_a-c_1}{q-p}=\parentheses*{c_a-c_1}\frac{z_{*}\parentheses*{y+c_a}}{\parentheses*{r_a-x}\parentheses*{c_a-c_1}} \Leftrightarrow \frac{v\parentheses*{t^1}}{z_{*}}=\frac{y+c_a}{r_a-x},
\end{align}
where the second equality follows from substituting the values of $p$ and $q$ from~\eqref{eq:pq}.
By the strict monotonicity of $v$ and~$\frac{v\parentheses*{z_{*}}}{z_{*}}=\frac{y+c_a}{r_a-x}$, we get that $t^1=z_{*}$ is the unique solution of~\eqref{eq:f}, as desired.

To see the ``moreover'' part of the theorem, we observe that the condition $r_a-v^{-1}\parentheses*{c_a}\geq \max\braces*{r_1-v^{-1}\parentheses*{c_1},w_0}$ is equivalent to the set of implementable utility profiles intersecting $F_a$. Once the intersection is nonempty --- our arguments above prove that the action $a$ is implementable. The discussion before the definition of $F_i$ implies the necessity of the condition on $s$; the sufficiency follows from the expected transfer under $t_{*}$ being $s=qz_{*}=r_a-x$, which can take every value in $r_a-\brackets*{\max\braces*{r_1-v^{-1}\parentheses*{c_1},w_0}, r_a-v^{-1}\parentheses*{c_a}}=\brackets*{v^{-1}\parentheses*{c_a}, r_a - \max\braces*{r_1-v^{-1}\parentheses*{c_1},w_0}}$.
\end{proof}

\section{Limiting the Misinformation of the Principal}
\label{sec:noise}
In some scenarios (see Subsection~\ref{sub:apply}), it is natural to assume that the principal necessarily has some relatively accurate information about the agent's action. We capture this intuition by the following stylized model. We set $k=n$ --- i.e., the number of signals equals the number of actions. We focus on a risk-neutral agent. It will be convenient for us in this section to impose the following genericity assumption: $w_i\neq w_j$ for $i\neq j$. Namely, all actions have different social welfare.

Fix some $0\leq d \leq n-1$ called the \emph{noise level}. We restrict the set of admissible information structures to be $\mathcal{I}\parentheses*{d}\subset \mathbb{R}^{n\times n}$, where 
\begin{align}\label{eq:id}
\begin{aligned}
    \mathcal{I}\parentheses*{d}:=\bigg\{ I\in \mathbb{R}^{n\times n}: \quad  &\forall a\in \brackets*{n}, \, \,  \sum_{j=1}^{n} I_{a,j}=1 , \\
    &\forall j\in \brackets*{n} \text{ s.t. }\exists a\in \brackets*{n}\text{ with }\ I_{a,j}>0 \quad \text{ (i.e., no dummy signals)} \\
    & \text{and } \, \, 
    I_{a,j}>0 \Rightarrow \absolute*{a-j}\leq d   \, \,  \qquad \qquad \text{ (i.e., bounded noise)} \bigg\}
\end{aligned}
\end{align}
Namely, for action $a\in \brackets*{n}$, the signal will be $j\in \brackets*{a-d,a+d}$ with probability 1. Moreover, for technical reasons, it is convenient for us to assume the \emph{no-dummy-signal} property as in the second line of Equation \eqref{eq:id}: Every signal is obtained with positive probability for some action. Similarly to the notion of implementability, we say that an action $a\in \brackets*{n}$ is \emph{$d$-implementable} if there exists $I\in \mathcal{I}\parentheses*{d}$ that implements the action~$a$.

In this section, we impose the limited liability assumption as in Section~\ref{sec:risk-neutral}-\ref{sec:risk-averse}. However, it turns out that the proofs of our results in this section (apart from Proposition~\ref{pro:hard} in Subsection~\ref{sub:negative}) do not rely on this assumption. Therefore, our results hold regardless of whether the limited liability is imposed or not.

The following theorem characterizes the set of $d$-implementable actions.

\begin{theorem}\label{thm:noise}
An action $a\in \brackets*{n}$ is $d$-implementable if and only if $w_a\geq \max\braces*{w_0, w_1}$ and:
\begin{align*}
   w_a \geq \max_{b\in \brackets*{n}} \min_{i\in \brackets*{n} \cap \brackets*{b-d,b+d}} w_i.
\end{align*}
\end{theorem}
Namely, an action is $d$-implementable if and only if its welfare (weakly) exceeds the minimal welfare among every $2d+1$ consequent actions, together with the least costly action and the default~action.
\begin{proof}
Assume first that $w_a\geq \max\braces*{w_0, w_1}$ and $w_a \geq \max_{b\in \brackets*{n}} \min_{i\in \brackets*{n} \cap \brackets*{b-d,b+d}} w_i$. We shall construct an information structure $I\in \mathcal{I}\parentheses*{d}$ that implements $a$. Let us first construct an information structure $I'$ that has the following three properties:
\begin{enumerate}
    \item For every action $i$, $I'$ sends a deterministic signal  (i.e., with probability $1$) $j\in \brackets*{i-d,i+d}$.
    \item $I'$ implements the action $a$.
    \item $I'$ yields principal's expected utility of $w_a$.
\end{enumerate}
However, the constructed $I'\notin \mathcal{I}\parentheses*{d}$ would violate the no-dummy-signal property. This issue can be overcome by adding a small $\epsilon>0$ probability to some carefully chosen actions.

The construction of $I'$ is as follows. Let $B=\braces*{b_1,...,b_m}\subset \brackets*{n}\setminus \braces*{a}$ with $b_1\leq b_2 \leq ... \leq b_m$ be the set of actions whose welfare is weakly below $w_a$, besides the action $a$ itself. Note that $b_1=1$ (as $w_a\geq w_1$) and $b_{i+1}-b_i\leq 2d+1$ (since otherwise $w_a>\min_{i'\in \brackets*{\parentheses*{b_i+d+1}-d,\parentheses*{b_i+d+1}+d}}w_{i'}$, which contradicts our condition for $b=b_i+d+1$). Moreover, $b_m\geq n-d$ --- indeed, otherwise $w_a>\min_{i'\in \brackets*{n-d,n}}w_{i'}$, which contradicts our condition for $b=n$.

The information structure $I'$ sends the same deterministic signal for all actions $b\in \brackets*{b_i,b_{i+1}-1}$. We set this signal to be $\lfloor \frac{b_i+b_{i+1}-1}{2} \rfloor$. Note that the condition $b_{i+1}-b_i\leq 2d+1$ ensures that, indeed, in each action $b$, the sent signal is in the interval $\brackets*{b-d,b+d}$. Similarly, for all actions $b\in \brackets*{b_m,n}$, we send the deterministic signal $n$.

For the information structure $I'$, the agent will never choose an action $b\notin B$, regardless of the contract. This is because action $i\in B$ maximizing $b_i$ subject to $b_i\leq b$ induces the same signal and is less costly. Therefore, actions $b\notin B$ can be ignored. We remain in a situation in which the actions in $B$ are fully revealed. Moreover, the action $a\in B$ has the maximal welfare. Therefore, the optimal contract will induce the action $a$ with a utility of $w_a$ for the principal.\footnote{Note that since $w_a\geq w_0$, a contract inducing action $0$ is not better for the principal than a contract extracting the full surplus from the action $a$.}

Now we carefully modify $I'$ to define $I\in \mathcal{I}\parentheses*{d}$ as follows. 
For every signal $j\in \brackets*{n}$ that is sent with probability $0$ for all actions, we replace the probability to $I_{b,j}=\epsilon>0$ for a single action $b \in \arg \min_{i\in \brackets*{b-d,b+d}\cap \brackets*{n}}$, which makes the signal $j$ being a non-dummy signal. Thereafter we adjust the $1$ weights of $I'$ to of $x\in \brackets*{1-2d\epsilon, 1}$ to make the matrix row-stochastic.

The key property of this additive $\epsilon$-adjustment is that the probability weights of $\epsilon$ are never added to actions $b\notin B$. Therefore, those actions can be ignored also in the adjusted information structure $I$.\footnote{This key property is crucial. E.g., had we non-carefully added the weight of $\epsilon$ to some action whose welfare is $w_b>w_a$, even though the probability $\epsilon$ is tiny, this would have caused a \emph{drastic} change in the optimal contract, which now might have implemented the action $b$; see Proposition~\ref{pro:prin-ut} below. In particular, it demonstrates that the principal's utility is non-continuous in the information structure and, therefore, standard continuity arguments with sufficiently small perturbations are insufficient for our proof.} In the remaining part of the principal-agent interaction, we are left with $\absolute*{B}$ actions, all of which are revealed with probability $1-O\parentheses*{\epsilon}$. Now, it is easy to verify that for a sufficiently small $\epsilon$, the induced action is the action $a$, whose welfare is strictly larger than the welfare of the other actions.

The proof of the converse part of the Theorem relies on the following general observation that might be of independent interest.

\begin{proposition}\label{pro:prin-ut}
Let $I\in \mathbb{R}^{n\times k}$ be an information structure (not necessarily in the restricted setting of the current section). For every signal $j\in \brackets*{k}$, let $\supp\parentheses*{j}=\braces*{a\in \brackets*{n}:I_{a,j}>0}\subseteq \brackets*{n}$ be the set of actions that induce the signal $j$ with a strictly positive probability. Then under an optimal contract, the principal's utility is at least $\min_{a\in \supp\parentheses*{j}} w_a$.
\end{proposition}

\begin{proof}
Let $t^*:= \min_{i\in \supp\parentheses*{j}} \frac{c_i}{I_{i,j}} $, and consider the contract that has a positive transfer of $t^*$ only in signal $j$ and $0$ transfer in all other signals.

By taking an action $a\notin \supp\parentheses*{j}$ the agent gets $0$ transfer -- and thus a strictly negative utility. By taking an action $a\in \supp\parentheses*{j}$ but $a \notin \arg  \min_{i\in \supp\parentheses*{j}} \frac{c_i}{I_{i,j}}$, the agent's expected utility under an optimal contract $t^*$ is $u^A_a=I_{a,j}t^*-c_a < I_{a,j} \frac{c_a}{I_{a,j}}-c_a=0$. By taking an action $a\in \arg  \min_{i\in \supp\parentheses*{j}} \frac{c_i}{I_{i,j}}$, the agent's expected utility is $u^A_a=I_{a,j}t^*-c_a=0$. 

Since the agent breaks ties in favor of the principal, she picks one of the actions $a \in  \arg  \min_{i\in \supp\parentheses*{j}} \frac{c_i}{I_{i,j}}$. The principal's utility for such a contract is $u^P_a=w_a-u^A_a=w_a\geq \min_{i\in \supp\parentheses*{j}} w_i$. In an optimal contract, the principal's utility might only increase.
\end{proof}

We turn to the proof of the converse part of the theorem. The necessity of the condition $w_a\geq \max\braces*{w_0,w_1}$ for the implementability of $a$ trivially follows from Theorem~\ref{thm:utility_neutral}. Indeed, if this condition is violated, then $a$ is not implementable. In particular, it is not $d$-implementable.

Let $a\in \brackets*{n}$ be an action with $w_a<\min_{i\in \brackets*{n} \cap \brackets*{b-d,b+d}} w_i$ for some $b\in \brackets*{n}$. We shall show that $a$ is not implementable. Indeed, by Proposition~\ref{pro:prin-ut} for the signal $j=b$ and by the no-dummy-signal assumption, we deduce that the agent's utility under an optimal contract must be at least $\min_{i\in \supp\parentheses*{b}}\geq \braces*{w_i}\min_{i\in \brackets*{b-d,b+d}} \braces*{w_i}$ (note that $\supp{\parentheses*{b}}\supseteq\brackets*{b-d,b+d}$). The maximal utility that the principal may get if the played action is $a$ is at most $u^P_a=w_a-u^A_a \leq w_a$. Therefore, $a$ is not implementable, as desired.
\end{proof}

\subsection{A Computational Hardness Result}
\label{sub:negative}
Inspired by the positive result in the previous subsection, we turn to studying a more general setting, still with a risk-neutral agent. Specifically, proceeding with the direct manager-employee example from Subsection~\ref{sub:apply}, we generalize the constrained design problem above by allowing an action to be mapped to a distant outcome, but with a bounded probability. This results in a strictly more general class of constraints in which the regulator is only allowed to use information structures with bounded probabilities of mapping actions to signals.\footnote{To be precise, here we do not impose the genericity assumption on the welfare for simplicity of presentation (we do not impose it also in our basic non-constrained model), but this issue can be easily resolved.} For this model, we show the following hardness result.

\begin{proposition}
\label{pro:hard}
Consider the information design problem with constraints specified by a matrix $C=\parentheses*{C_{i,j}}\in\brackets*{0,1}^{n\times k}$ such that only information structures that satisfy the condition $I_{i,j}\le C_{i,j}$ for each action $1\le i\le n$ and signal $1\le j\le k$ are allowed. Then deciding whether any given action $a$ is implementable or not is NP-complete.
\end{proposition}

Note that by setting trivial bounds from $\braces*{0,1}$ on the probabilities, one can get both the unconstrained information design principal-agent problem, and the constrained version introduced at the beginning of this section. 

\begin{proof}[Proof of Proposition~\ref{pro:hard}]
    It is not hard to see that the problem belongs to NP, as the implementing information structure has a polynomial-sized representation. To prove NP-hardness, we reduce from the set cover problem. Consider a set cover instance with a universe $U=\braces*{0,\ldots,s-1}$ and given sets $S_1,\ldots,S_t$, and assume that we wish to decide whether it is possible to cover the universe using all but $l$ sets. Let us introduce the following constrained instance of the information design principal-agent problem. For each set $S_j$, define an action $a_j$ with cost $1$ and reward $1$ and a signal $\sigma_j$. For every element $i\in U$, introduce an action $b_i$ with cost $0$ and reward $1$. Add a further action $a_0$ with cost $\frac{l}{t}$ and reward $1+\frac{l}{t}$, and take $a=a_0$. Set $C_{a_j,\sigma_j}=1$ for $j=1,\ldots,t$, $C_{b_i,\sigma_j}=1$ for $i\in S_j$, $C_{a_0,\sigma_j}=\frac{1}{t}$ for $j=1,\ldots,t$, and let all the other entries in $C$ be $0$. Namely, an action $b_i$ representing some element $i\in U$ can only be mapped to signals corresponding to sets containing $i$ (with any probability), an action $a_j$ with $j>0$ must be deterministically mapped to the signal $\sigma_j$, and $a_0$ can be mapped to any signal with probability at most $1/t$.
 
    If there exist $t-l$ sets out of $S_1,\ldots,S_t$, with their set of indices being some $J$, covering the entire universe $U$, let us consider the information structure distributing the action $a_0$ uniformly among the $t$ signals, mapping each action $b_i$ deterministically to some $\sigma_j$ with $j\in J$ s.t.~$i\in S_j$, and mapping each $a_j$ with $j>0$ deterministically to $\sigma_j$. Then the contract paying $1$ for each signal with an index outside $J$ gives the agent expected utility of $l\cdot \frac{1}{t}-\frac{l}{t}=0$ for taking action $a_0$, and no action gives the agent a strictly positive expected utility. When the agent takes the action $a_0$, the principal gets the expected utility of $\parentheses*{1+\frac{l}{t}}-l\cdot \frac{1}{t}=1$. As no action has social welfare larger than~$1$, this situation is optimal for the principal.
 
    Assume now that under some information structure $I$, there exists an optimal contract inducing the action $a_0$. Note that under the contract with the constant monetary transfer of $0$, the action $b_0$ is (weakly) optimal for the agent, yielding the principal expected utility of $1$. Thus, if $a_0$ is implementable, then in the implementing optimal contract the principal must extract the full social welfare. In particular, under the implementing information structure, in any optimal contract such that the agent takes action $a_0$:~(i)~no positive monetary transfer can be made for any signal that occurs with a positive probability for some action $b_i$ (as such an action costs $0$ to the agent);~(ii)~no monetary transfers larger than $1$ can occur (as $a_j$ with $j>0$ costs $1$ to the agent, and it is deterministically mapped to $\sigma_j$).

    Let $J$ be the set of the indices of all signals to which at least one action $b_i$ is mapped with a positive probability. By definition of the matrix $C$, $\braces*{S_j}_{j\in J}$ is a cover of $U$. If $\absolute*{J}>t-l$, then (i) implies that $a_0$ is mapped with a positive probability under $I$ to strictly less than $l$ signals with a positive monetary transfer. Therefore, by (ii) and the definition of $C$, we get that the agent's expected utility upon taking $a_0$ is strictly smaller than $l\cdot \frac{1}{t}-\frac{l}{t}=0$, meaning that the contract would not induce the action~$a_0$.
\end{proof}

\section{Additional Related Literature}
\label{sec:related}
\textbf{Closely related works. } Among the papers related to ours, we should mention~\citet{garrett2023optimal}, who analyze the implementability of principal-agent utility profiles, but in a setting different from ours. Specifically, they assume that the social planner can pick any bounded mapping from a continuum-sized set of outcomes to the agent's costs. Their characterization turns out to differ from our results; in particular, their set of inducible utility profiles is convex, which is not true in the framework we study. \citet{mahzoon2023indicator}~consider a framework similar to ours, with the social planner's incentives being perfectly aligned with those of the agent under the limited liability assumption (as we do in Example~\ref{ex:agent-max} for a risk-neutral agent). They apply a geometric approach to analyze the effect of properties of the information structure on implementable utility profiles. A follow-up work to our paper~\cite{lin2023principal} gives explicit expressions for the information structure maximizing Utilitarian social welfare and Nash social welfare~\cite{kaneko1979nash}, as well as for the corresponding utility profile and the action taken by the agent. Similarly to us, they study both a risk-neutral and a risk-averse agent case. These results can be viewed as an extension of our Example~\ref{ex:strategic}. Compared to these papers, we study the implementability of actions and utility profiles in a general setting. While our focus is not on optimization issues, the follow-up work shows the relevance of our results to optimization.

\textbf{The principal-agent problem. } Starting with the seminal papers of~\citet{ross1973economic,mirrlees1976optimal,holmstrom1979moral}, there has been a vast literature on the principal-agent problem. In particular, the dependence of the principal's utility on the information structure has been discussed --- starting with~\citet{gjesdal1982information}, who studies a variation of the principal-agent problem he refers to as the \emph{agency information problem}. The model in that paper is more general than ours --- it assumes that not only the outcome is not observed by the principal, but also the utilities of both the principal and the agent depend on an unknown state. The paper ranks the information structures according to the principal's preferences, providing a characterization for the ranking that generalizes the Blackwell ordering~\cite{blackwell1950comparison, Blackwell53}.~\citet{kim1995efficiency}~shows such a characterization based on second-order stochastic dominance, and~\citet{demougin2001ranking} introduce an equivalent criterion.

\citet{radner1984nonconcavity} and~\citet{singh1985monitoring}~suggest mild assumptions on the setting under which the gross marginal value of the information at the point of no-revelation is, respectively, non-positive and zero;~\citet{demougin2001monitoring} compare different monitoring policies when both the principal and the agent are risk-neutral under the limited liability constraint;~and~\citet{namazi2013role} compares information structures under different kinds of uncertainty about the outcome. However, our focus is different from these papers --- we are interested in the interplay of the principal's \emph{and agent's}~utilities.

\citet{milgrom1981good}~introduces the first-order stochastic-dominance-based notion of ``favorableness'' of the information revealed to the principal and studies its effect on the expected transfer --- focusing not only on the principal's but also on the agent's point of view.~\citet{jacobides2001information,silvers2012value,chaigneau2018does} further study the effect on both for the principal's and the agent's utilities by improving the informativeness of the information structure.~\citet{sobel1993information} explores the effect of the agent's exogenous information on the principal's expected utility. Our results for bounded noise on the signals also address the effect of informativeness. While these works focus on fixed information structures, we study the case in which the information structure is not specified by the model; rather, it can be adjusted by a regulator to serve its purpose.

There are further variations of the principal-agent problem with partially-observed outcomes. An important example is studying the effect of the principal's supervision policy on the agent with some information being non-contractible. The effect of changing the performance measures of a contract on the equilibrium payoffs has already been discussed in~\cite{holmstrom1979moral}. This question has been further studied for linear contracts with non-contractible principal's payoff~\cite{baker1992incentive}, as well as for contracts in a multi-task model encompassing the tradeoff between risk and distortion~\cite{baker2002distortion}.~\citet{zhu2018better}~shows that improving the informativeness of a supervision policy in an organization might decrease the organization's payoff. A work by~\citet{abreu1990toward} studies the effect of the informativeness of the monitoring policy on the equilibrium payoffs in a different framework of repeated games. Unlike these papers, we mostly consider all possible implementable utility profiles under \emph{some} supervision policy, rather than studying the effect of the revealed information on the payoffs. An exception is Section~\ref{sec:noise}, which considers the set of implementable actions as a function of the noise level of the signal. In our setting, adding more noise to the signals might only \emph{decrease} the worst-case principal's utility over the set of all implementable~actions.

As part of our Theorem~\ref{thm:no_ll}, we show that without the limited liability assumption, the agent's expected utility under an optimal contract must be zero, whether the agent is risk-neutral or risk-averse. It can also be immediately deduced from the results of~\citet{grossman1992analysis} that apply to additively separable agent's utility. Among further related models we should mention~\cite{hoffmann2021only} --- assuming that both observing the outcomes and making the transfer to the agent can be delayed, the paper focuses on the optimal timing of the payment.

\textbf{Bayesian persuasion. } Our setting is conceptually related to the extensive recent literature on information design in Sender-Receiver models, e.g., the seminal work of~\citet{kamenica2011bayesian} on Bayesian persuasion (for some recent surveys, see~\cite{dughmi2017algorithmic,kamenica2019bayesian, bergemann2019information}). A particularly relevant branch of the Bayesian persuasion literature focuses on moral hazard problems.~\citet{boleslavsky2018bayesian}~introduce a three-party interaction between a sender a receiver and an agent.~\citet{boleslavsky2018bayesian}~assume that the prior distribution is chosen endogenously. Specifically, the prior is specified by the agent's level of effort incentivized by the sender's information structure.~\citet{zapechelnyuk2020optimal}~study another type of moral hazard in Bayesian persuasion in which the state of nature, representing a product quality level, is chosen by a strategic agent (producer). The producer further chooses a publicly-known price for the product. The sender (regulator) chooses an information structure specifying the information revealed to the receiver (buyer) about the product quality. He shows that choosing the information structure can be considered as a \emph{delegation problem}~\cite{holmstrom1980theory} (for application to Bayesian persuasion, see, e.g.,~\cite{kolotilin2019persuasion}). That is, the regulator can w.l.o.g.~just restrict the set of quality levels that the producer may use. We use a similar technique in Theorem~\ref{thm:utility_neutral},~\ref{thm:utility_averse} and~\ref{thm:no_ll} proofs --- specifically, we provide an information structure that reduces the set of relevant actions. Further works combining Bayesian persuasion with moral hazard include~\citet{rodina2017information} (Bayesian persuasion affected by the agent's career concerns~\cite{holmstrom1999managerial}) and~\citet{kwak2022optimal} (a binary-action binary-outcome principal-agent setting in which the outcome depends on a state unknown to the~agent).

Despite the conceptual similarity, there are several fundamental differences in designing information about a state of nature in the Bayesian persuasion setting and about a strategic agent's action in the principal-agent setting. Information design cannot affect a nature's choice of a state, but it does affect the agent's chosen action. In particular, Bayesian persuasion models typically assume a common prior distribution over the states of nature, while there is no prior on the agent's action in the principal-agent setting. This crucial difference leads to an interesting information design problem, in which the designed information structure affects the receiver's action (i.e., the principal's contract) via solving an induced optimal contract problem, as opposed to picking the action with the best expected posterior utility as in classic information design. This angle may spur more research on non-Bayesian-persuasion-style information design.

\section{Conclusions and Future Work}
\label{sec:con}
This paper analyzes the implementability of actions and utility profiles in the principal-agent model under certain information structures. We provide a clean necessity and sufficiency characterization for implementability both for a risk-neutral and a risk-averse agent, with and without the limited liability assumption. It turns out that a combination of conditions that are naturally necessary for implementability is also sufficient. We further extend this study of implementability to situations with constrained information structures. Similarly, we provide a tight characterization of implementable actions when the noise added to the agent's action is bounded by some distance $d$. As a corollary of our result,  the total welfare degrades gracefully as the noise level $d$~increases. In contrast, we prove that deciding whether an action is implementable is NP-complete for a more general class of constrained information structures.

Our model and results give rise to many interesting open research questions. First, one can study other notions of ``bounded noise''. For example, instead of bounded distance, another natural notion is \emph{expected} noise bounded by some $d$. Secondly, our bounded noise model can also be viewed as a \emph{bipartite constraint graph} in which actions are on one side, whereas signals are on the other side. Bounded noise corresponds to a specific bipartite constraint graph. Thus, an interesting research direction is to consider similar questions, as well as computational questions, for a general bipartite constraint graph on the information structure. Another intriguing bounded setting requires the admissible information structures to be Blackwell-dominated by a given exogenous information structure; this restriction naturally represents a principal observing a noisy signal about the outcome specified by the exogenous information structure.

Thirdly, and further zooming out, one interesting future research direction is to extend our model to principal-agent problems with a much richer structure than the basic one considered in this paper. For example, one can ask similar questions for combinatorial contracts~\cite{dutting2022combinatorial} in which an agent can take multiple actions, and typed contract design in which the agent has private information as well, modeled as a private agent's type~\cite{darrough1986moral,bajari2014moral}. The latter question is particularly intriguing since it leads to an interplay of information design, adverse selection due to the private agent's type, and moral hazard due to unobservable agent's~actions.

Finally, there are various further algorithmic questions not studied yet. For instance, in many real-life applications, the regulator may be constrained by guaranteeing a certain minimal welfare share for both the principal and agent (e.g., for fairness reasons). One algorithmic question is how to compute an information structure such that the equilibrium of the resultant principal-agent problem maximizes welfare subject to the welfare share constraints. Furthermore, in our characterization of constrained information structures, while we can construct information structures that achieve any implementable action, these information structures unfairly favor the principal and lead to a $0$ surplus for the agent. While this serves our purpose of characterizing the boundary of implementable actions, the construction may not be practically useful. This leads to an open question of computing an information structure satisfying some practical constraints that implements a given action.

\paragraph{\bf Acknowledgements.} Yakov Babichenko is supported by the Binational Science Foundation BSF grant No.~2018397. This work is funded by the European Union (ERC, ALGOCONTRACT, 101077862, PI: Inbal Talgam-Cohen). Inbal Talgam-Cohen is supported by a Google Research Scholar award and by the Israel Science Foundation grant No.~336/18. Yakov and Inbal are supported by the Binational Science Foundation BSF grant No.~2021680. Haifeng Xu is supported by NSF award No.~CCF-2303372, Army Research Office Award No.~W911NF-23-1-0030 and Office of Naval Research Award No.~N00014-23-1-2802. Konstantin Zabarnyi is supported by a PBC scholarship for Ph.D. students in data science. The authors are grateful to anonymous reviewers for their helpful remarks.

\bibliographystyle{plainnat}
\bibliography{main}

\appendix
\section{Dropping the Limited Liability Assumption}
\label{ap:further}
In this appendix, we discuss the implementable utility profiles and action-transfer pairs after dropping the limited liability assumption --- that is, \emph{the transfers might be negative}. Formally, the principal may select any vector $\parentheses*{t_1,...,t_n}\in\mathbb{R}^n$ of transfers, rather than only $\parentheses*{t_1,...,t_n}\in\mathbb{R}_{\geq 0}^n$. When the agent is risk-averse, we extend the definition of the von Neumann-Morgenstern utility function to the whole real line. That is, we take a concave and strictly increasing $v:\mathbb{R}\to \mathbb{R}$ with $v\parentheses*{0}=0$ and $\lim_{z\to \infty}\frac{v\parentheses*{z}}{z}=0$. The agent's utility at action $i\in \brackets*{n}$ is still $u^A_i=\mathbb{E}_{j\sim I_i} \brackets*{v\parentheses*{t^j}}-c_i$.

We shall prove that while dropping the limited liability assumption does not affect the set of implementable actions, the only implementable agent's utility is $0$. Concretely, the set of no-limited-liability-implementable utility profiles is the intersection of the set of limit-liability-implementable utility profiles (which have been characterized in Theorems~\ref{thm:utility_neutral} and~\ref{thm:utility_averse}) with the line $u^A=0$; see Figures~\ref{fig:1} and~\ref{fig:3}.
\begin{theorem}\label{thm:no_ll}
Both for the risk-neutral agent case and the risk-averse agent case, the same actions are implementable with and without the limited liability assumption (as specified by Theorems~\ref{thm:utility_neutral} and~\ref{thm:utility_averse}, respectively). For an implementable action $a$, $\parentheses*{a,s}$ is implementable if and only if :
\begin{itemize}
    \item $s=c_a$ for a risk-neutral agent,
    \item $s\in\brackets*{v^{-1}\parentheses*{c_a}, r_a - \max\braces*{r_1-v^{-1}\parentheses*{c_1},w_0}}$ for a risk-averse agent.
\end{itemize}
Furthermore, the set of implementable utility profiles is:
\begin{itemize}
    \item $\braces*{\parentheses*{w_a,0}: a \text{ is implementable}}$ for a risk-neutral agent, and 
    \item $\braces*{\parentheses*{x,0}: a \text{ is implementable and } x\in\brackets*{\max\braces*{r_1-v^{-1}\parentheses*{c_1},w_0},r_a-v^{-1}\parentheses*{c_a}}}$ for a risk-averse~agent.
\end{itemize}
Moreover, all implementable actions and utility profiles are implementable using just two signals.
\end{theorem}
\begin{proof}
The necessity of the conditions on implementable actions follows exactly as in Theorems~\ref{thm:utility_neutral} and~\ref{thm:utility_averse}. The necessity of $y=0$ in any implementable utility profile $\parentheses*{x,y} \in F$ follows (in both settings) from the observation that if a contract $t$ induces strictly positive expected agent's utility --- subtracting a tiny number from each of its coordinates does not change the induced action and decreases the expected transfer; therefore, $t$ cannot be optimal. The necessity of the condition on $x$ for an implementable $\parentheses*{x,0}\in F$ follows immediately from the definition of $F$ (and the principal's alternative of inducing action $1$ in the risk-averse~case).

It remains to prove the sufficiency of the presented conditions for implementability of actions and utility profiles (the conditions on $s$ follow as in Theorem~\ref{thm:utility_neutral} and~\ref{thm:utility_averse} proofs). We shall only do this for a risk-averse agent; the risk-neutral case is similar and even simpler to prove.

Indeed, suppose that an action $a$ satisfies $r_a-v^{-1}\parentheses*{c_a}\geq \max\braces*{r_1-v^{-1}\parentheses*{c_1},w_0}$, and $x$ satisfies $x\in\brackets*{\max\braces*{r_1-v^{-1}\parentheses*{c_1},w_0},r_a-v^{-1}\parentheses*{c_a}}$. Define $p,q,d_0,z_{*}$ as in Theorem's~\ref{thm:utility_averse} proof, taking $y=0$. Similarly to that proof, we shall show that the action $a$ and the utility profile $\parentheses*{x,0}$ can be implemented by a binary-signal information structure $I$ with $I_i=\parentheses*{q,1-q}$ for $i=a$ and any $i\in \brackets*{n}$ s.t.~$c_i>c_a$, and $I_i=\parentheses*{p,1-p}$ for the remaining $i\in \brackets*{n}$.

As in Theorem~\ref{thm:utility_averse} proof --- it is enough to consider the case $w_1\geq w_0$; moreover, for all contracts, the agent prefers action $1$ over the actions $i\notin \braces*{0,a}$. Moreover, the agent's utilities at actions $0$, $1$, and $a$ for the contract $t_*=\parentheses*{z_{*},0}$ are $0$; thus, by the tiebreak rule, she takes action $a$ under $t_*$, yielding utility of $x$ to the principal. It remains to show that $t_*$ is the optimal contract for the principal. To this end --- it is enough to prove that there exists an optimal contract $t=\parentheses*{t^1,t^2}$ with $t^2\geq 0$.\footnote{From this point, we can finish as in Theorem~\ref{thm:utility_averse} proof.}

Indeed, consider an optimal contract $t=\parentheses*{t^1,t^2}$ with $t^2<0$. As $t$ is optimal, we have $u^A_a\parentheses*{t}\geq u^A_1\parentheses*{t}$ and $u^A_a\parentheses*{t}\geq 0$. Define a contract
\begin{align*}
    t_{\epsilon}=\parentheses*{t_{\epsilon}^1,t_{\epsilon}^2}:=\parentheses*{t^1-\epsilon,t^2+\frac{\epsilon q}{1-q}}
\end{align*}
for some $\epsilon\geq 0$ that satisfies\footnote{Here we assume $q<1$; for $q=1$ the proof is straightforward.}
\begin{align*}
 t^1-\epsilon \geq t^2+\frac{\epsilon q}{1-q}.  
\end{align*}
 As $t^1 > t^2$ and $q t^1 +\parentheses*{1-q} t^2 = q t_{\epsilon}^1 +\parentheses*{1-q} t_{\epsilon}^2$, we get by the concavity of $v$ that:
$$0\leq u^A_a\parentheses*{t} =qv\parentheses*{t^1}+\parentheses*{1-q}v\parentheses*{t^2}-c_a\leq qv\parentheses*{t_{\epsilon}^1}+\parentheses*{1-q}v\parentheses*{t_{\epsilon}^2}-c_a=u^A_a\parentheses*{t_{\epsilon}}.$$
Moreover, by definition of $t_{\epsilon}$, we have $u^P_a\parentheses*{t_{\epsilon}}=u^P_a\parentheses*{t}$.

As taking $\epsilon\geq 0$ s.t.~$t^1-\epsilon=t^2+\frac{\epsilon q}{1-q}$ reverses the sign of the inequality in~\eqref{eq:a1} (for $t_{\epsilon}$ serving as $t$) --- the continuity of $v$ implies that there exists $\epsilon\geq 0$ s.t.~the inequality in~\eqref{eq:a1} becomes equality for $t_{\epsilon}$. This $t_{\epsilon}$ leads to action $a$, as $u^A_a\parentheses*{t_{\epsilon}}\geq u^A_1\parentheses*{t_{\epsilon}}$ and $u^A_a\parentheses*{t_{\epsilon}}\geq 0$. Moreover, $t_{\epsilon}$ is optimal, since $t$ is optimal, both contracts induce action $a$, and $u^P_a\parentheses*{t_{\epsilon}}=u^P_a\parentheses*{t}$. However, as we saw in Theorem~\ref{thm:utility_averse} proof, the tightness of the inequality in~\eqref{eq:a1} implies the tightness of the first inequality in~\eqref{eq:a0}. Therefore, $0\leq u^A_a\parentheses*{t_{\epsilon}}=v\parentheses*{t_{\epsilon}^2}$, which implies $t_{\epsilon}^2\geq 0$. Thus, $t_{\epsilon}$ is the desired optimal~contract.
\end{proof}

\end{document}